\documentclass[11pt]{article}
\usepackage{amsmath, amssymb, amsfonts, amsthm}
\usepackage{epsfig}

\newtheorem{theorem}{Theorem}[section]

\newtheorem{lemma}[theorem]{Lemma}
\newtheorem{definition}[theorem]{Definition}

\newcommand{\be}{\begin{equation}}
\newcommand{\ee}{\end{equation}}
\newcommand{\bey}{\begin{eqnarray}}
\newcommand{\eey}{\end{eqnarray}}

\newcommand{\E}{{\mathbb E }}
\renewcommand{\P}{{\mathbb P}}

\newcommand{\eps}{\varepsilon}

\newcommand{\bw}{{\bf w}}

\newcommand{\bv}{{\bf v}}
\newcommand{\bu}{{\bf u}}
\newcommand{\bba}{{\bf a}}
\newcommand{\bb}{{\bf b}}

\newcommand{\bz}{{\bf z}}

\newcommand{\bR}{{\mathbb R}}
\newcommand{\bC}{{\mathbb C}}
\newcommand{\bN}{{\mathbb N}}

\newcommand{\tr}{\mbox{Tr}}

\newcommand{\const}{\mathrm{const}}

\newcommand{\cN}{{\cal N}}

\input epsf

\newcommand{\donothing}[1]{}

\begin{document}

\title{A Wegner estimate for Wigner matrices}

\author{Anna Maltsev and Benjamin Schlein\thanks{Partially supported
by an ERC Starting Grant} \\
\\
Institute of Applied Mathematics, University of Bonn, \\
Endenicher Allee 60, 53115 Bonn}

\maketitle

\begin{abstract}
In the first part of these notes, we review some of the recent developments in the study of the spectral properties of Wigner matrices. In the second part, we present a new proof of a Wegner estimate for the eigenvalues of a large class of Wigner matrices. The Wegner estimate gives an upper bound for the probability to find an eigenvalue in an interval $I$, proportional to the size $|I|$ of the interval.
\end{abstract}

\section{Introduction}

The general goal of Random Matrix Theory consists in establishing statistical properties of the eigenvalues of $N \times N$ matrices whose entries are random variables with a given probability law, in the limit of large $N$. In these notes, we will focus on so called Wigner matrices whose entries are, up to the symmetry constraints, independent and identically distributed random variables.

Wigner matrices were originally introduced by Wigner to describe the excitation spectrum of heavy nuclei. Wigner's intuition was as follows: the Hamilton operator of a complex system (such as a heavy nucleus) depends on so many degrees of freedom that it is essentially impossible to write it down precisely. Hence, it makes sense to assume the matrix elements of the Hamilton operator to be random variables, and to study properties of the spectrum which hold for most realizations of the randomness. Remarkably, it turned out that the distribution of the excitation energies of heavy nuclei is indeed well approximated by the distribution of the eigenvalues of Wigner matrices. Today, Wigner matrices have been linked to several other branches of mathematics and physics. The distribution of the eigenvalues of random Schr\"odinger operators in the metallic phase (where eigenvectors are delocalized), for example, is expected to be close to the one  observed in hermitian ensembles of Wigner matrices. Similarly, the spectrum of the Laplace operator on domains with chaotic classical trajectories is expected to share  several properties with the spectrum of real symmetric Wigner matrices.

The success of Wigner's idea can be understood as a consequence of universality. In vague terms, universality states that the distribution of the eigenvalues of disordered (or chaotic) systems depends on the underlying symmetry but is otherwise independent of further details. This concept is very general, and, from the mathematical point of view, its validity remains a mystery. Nevertheless, in the last years a lot of progress was made in the mathematical analysis of Wigner matrices and, at least in this context, the emergence of universality has been by now understood.

These notes are organized as follows. In Section \ref{sec:review} we first give the precise definition of the ensembles that we are going to study. Then, we briefly review some of the results on Wigner matrices obtained in the last few years. Finally, in Section \ref{sec:weg}, we present a new Wegner estimate for the eigenvalues of a large class of Wigner matrices.

\section{Some spectral properties of Wigner matrices}
\label{sec:review}

To simplify the presentation, we will restrict our attention to ensembles of hermitian Wigner matrices. However, most of the results that we are going to present extend also to ensembles with different symmetry (real symmetric and quaternion hermitian ensembles).

\begin{definition}\label{def} An ensemble of {\it Hermitian Wigner matrices} consists of $N \times N$ matrices $H = (h_{jk})_{1\leq j,k \leq N}$, with \[ \begin{split}  h_{jk} &= \frac{1}{\sqrt{N}} (x_{jk} + i y_{jk})  \qquad \text{for } 1 \leq j <k \leq N \\
h_{jk} & = \overline{h}_{kj} \qquad \qquad \qquad \quad\text{for } 1\leq k < j \leq N \\
h_{jj} & = \frac{1}{\sqrt{N}} x_{jj} \qquad \qquad \quad \text{for } 1 \leq j \leq N \end{split} \]
where $\{ x_{jk}, y_{jk}, x_{jj} \}_{1 \leq j < k \leq N}$ is a collection of $N^2$
independent real random variables. The (real and imaginary parts of the) off-diagonal entries $\{ x_{jk} , y_{jk} \}_{1\leq j < k \leq N}$ have a common distribution with \[ \E \, x_{jk} = 0 \qquad \text{and } \E \, x_{jk}^2 = \frac{1}{2} \, .  \]
Also the diagonal entries $\{ x_{jj} \}_{j=1}^N$ have a common distribution with $\E \, x_{jj} = 0$ and $\E \, x_{jj}^2 = 1$. For technical reasons, we also assume the entries to decay sufficiently fast at infinity, in the sense that
\[ \E \, e^{\nu |x_{ij}|^2} < \infty, \quad \text{ and }  \quad \E \, e^{\nu |x_{ii}|^2} < \infty \,  \]
for some $\nu >0$.
\end{definition}

Observe that the entries $h_{jk}$ scale, by definition, with the dimension $N$ of the matrix. We choose this scaling so that, in the limit of large $N$, all eigenvalues of $H$ remain of order one. To show that this is indeed the right scaling, consider the trace of $H^2$. On the one hand,
\begin{equation}\label{eq:trH2} \E \, \tr \, H^2 = \E \sum_{i,j=1}^N |h_{ij}|^2 = N^2 \, \E |h_{12}|^2 \end{equation}
since all entries have the same distribution. On the other hand, if $\mu_1, \dots , \mu_N$ denote the $N$ eigenvalues of $H$, we have
\[ \E \, \tr\, H^2 = \E \, \sum_{\alpha=1}^N \mu^2_\alpha. \]
If all eigenvalues are of order one in the limit $N \to \infty$, the r.h.s. is a quantity of the order $N$. Comparing with (\ref{eq:trH2}), it is clear that this is only possible if $\E \, |h_{12}|^2$ is of the order $N^{-1}$; this explains the scaling of the matrix entries introduced in Definition \ref{def}.

The best known ensemble of hermitian Wigner matrices is the so called Gaussian Unitary Ensemble (GUE) which is characterized by the further assumption that the random variables $\{ x_{jk}, y_{jk}, x_{jj} \}$ are Gaussian. It turns out that GUE is the only ensemble of hermitian Wigner matrices which is invariant w.r.t. unitary conjugation; if $H$ is a GUE matrix, also $U H U^*$ is a GUE matrix, for every fixed unitary matrix $U$. Because of the unitary invariance, for GUE it is possible to compute explicitly the joint probability density of the $N$ eigenvalues; it is given by
\begin{equation}\label{eq:GUE} p_{\text{GUE}} (\mu_1, \dots , \mu_N) = \const \cdot \prod_{i<j}^N (\mu_i - \mu_j)^2 e^{-\frac{N}{2} \sum_{j=1}^N \mu_j^2}  \, . \end{equation}
Note that GUE is the only ensemble of hermitian Wigner matrices for which an explicit expression for the joint probability density function of the $N$ eigenvalues is available.
Since we are interested here in establishing statistical properties of the spectrum of Wigner matrices which hold true independently of the specific choice of the probability law for the entries, we will not make use the expression~(\ref{eq:GUE}).

The first rigorous result in random matrix theory was the proof obtained by Wigner in \cite{W} of the convergence of the density of states (density of eigenvalues) to the semicircle law in the limit of large $N$. For $a < b$, let $\cN [a;b]$ denote the number of eigenvalues of a Wigner matrix $H$ in the interval $[a;b]$. The density of states in the interval $[a;b]$ is defined as $\cN [a;b] / N (b-a)$ (because of the scaling of the matrix entries, the typical distance between neighboring eigenvalues is of the order $1/N$; this is the reason why, in order to obtain a quantity of order one, we have to divide the density of states by $N$). In \cite{W}, Wigner proved that, for any fixed $a < b$ and $\delta >0$,
\[ \lim_{N \to \infty} \P \left( \left| \frac{\cN [a;b]}{N|b-a|} - \int_a^b \rho_{\text{sc}} (s) ds \right| \geq \delta \right) = 0 \, . \]
The limiting density of states is given by the famous semicircle law
\begin{equation}
\label{eq:sc}
\rho_{\text{sc}} (s) = \left\{ \begin{array}{ll} \frac{1}{2\pi} \sqrt{1- \frac{E^2}{4}}, \quad &\text{if } |E| \leq 2 \\ 0 \quad &\text{if } |E| > 2 \end{array} \right. \, . \end{equation}
Observe, in particular, that the semicircle law is independent of the choice of the probability law for the entries of $H$.

\subsection{Semicircle law on microscopic intervals}
\label{sec:micro}

It is important to remark that Wigner's result concerns the density of states on intervals whose size is independent of $N$. Such intervals contain, typically, a non trivial fraction of the total number of eigenvalues $N$. We say, for this reason, that these intervals are macroscopic. It seems then natural to ask what happens if one considers smaller intervals, namely intervals whose size shrinks down to zero as $N \to \infty$. These intervals will not contain order $N$ eigenvalues, so they will not be macroscopic, but as long as they contain a large number of eigenvalues in the limit of large $N$, it turns out that one still has convergence to the semicircle law. This is the content of the next theorem, which was first proven in \cite{ESY3}, using also partial results from \cite{ESY1,ESY2}.
\begin{theorem}\label{thm:micro}
Consider an ensemble of hermitian Wigner matrices as in Def. \ref{def}. Let $|E|<2$. Then
\begin{equation}\label{eq:microsc} \lim_{K \to \infty} \lim_{N \to \infty} \P \left( \left| \frac{\cN \left[E-\frac{K}{2N} ; E+\frac{K}{2N} \right]}{K} - \rho_{\text{sc}} (E) \right| \geq \delta \right) = 0 \, . \end{equation}
\end{theorem}
In contrast with Wigner's original result, this theorem establishes the convergence of the density of states to the semicircle law on microscopic intervals, that is on intervals containing, typically, a constant ($N$ independent) number of eigenvalues.
{F}rom the convergence on the microscopic scale, it is easy to conclude convergence on arbitrary intermediate scales; for any $|E| <2$ and any sequence $\eta (N) >0$ such that $\eta (N) \to 0$ and $N \eta (N) \to \infty$ as $N \to \infty$ we have
\[ \lim_{N \to \infty} \P \left( \left| \frac{\cN \left[E-\frac{\eta(N)}{2} ; E+\frac{\eta (N)}{2} \right]}{N\eta (N)} - \rho_{\text{sc}} (E) \right| \geq \delta \right) = 0\, . \]
Note that, if $\eta (N) \lesssim 1/N$, the fluctuations of the density of states are certainly important, and one cannot expect convergence in probability.

\subsection{Delocalization of eigenvectors}
\label{sec:deloc}

As a simple application of the convergence to the semicircle law on microscopic scales, one can show the complete delocalization of the eigenvectors of Wigner matrices.
Let $\bv \in \bC^N$ with $\| \bv \|_2 =1$. The vector $\bv$ is said to be completely localized if one of its component has size one, and all other components vanish. On the other hand, $\bv$ is called completely delocalized, if all its components have the same size (namely $N^{-1/2}$). In order to distinguish localized from delocalized vectors, one can compute the $\ell^p$ norm, for $p >2$. If $\bv$ is completely localized,
$\| \bv \|_p =1$ for all $p \geq 2$ and for all $N \in \bN$. If $\bv$ is completely delocalized, $\| \bv \|_p = N^{-\frac{1}{2} + \frac{1}{p}}$ and converges to zero, as $N \to \infty$. The next theorem was proven in \cite{ESY3}, extending results from \cite{ESY1,ESY2}.
\begin{theorem}\label{thm:deloc}
Consider an ensemble of hermitian Wigner matrices as in Def. \ref{def}. Fix $|E|<2$, $K >0$ and $2 < p < \infty$. Then
\begin{equation}\label{eq:deloc1} \lim_{M \to \infty} \lim_{N \to \infty}
\P \left( \exists \bv : H \bv = \mu \bv, |\mu - E| \leq \frac{K}{N}, \| \bv \|_2 = 1, \|\bv \|_p \geq M N^{-\frac{1}{2}+ \frac{1}{p}} \right) = 0 \, . \end{equation}
\end{theorem}
Eq. (\ref{eq:deloc1}) shows the complete delocalization of the eigenvectors of Wigner matrices; up to constants, all components of eigenvectors have the same size.

\subsection{Universality of Wigner matrices}
\label{sec:uni}

For Wigner matrices, universality refers to the fact that the local eigenvalue correlations depend on the symmetry of the ensemble but, otherwise, they are independent of the probability law of the entries.

One distinguishes, typically, between universality at the edge and in the bulk of the spectrum. For hermitian Wigner matrices, the local statistics at the edges are described by the Tracy-Widom distribution; see \cite{TW,S}. Here, we restrict our attention to bulk universality.

Let $p_N (\mu_1, \dots , \mu_N)$ be the joint probability density function of the $N$ (unordered) eigenvalues of a hermitian Wigner matrix $H$. For any $k = 1, 2, \dots, N$, we define the $k$ point correlation function
\begin{equation}\label{eq:kpoint} p^{(k)}_N (\mu_1, \dots , \mu_k) = \int d\mu_{k+1} \dots d\mu_{N} \, p_N (\mu_1, \dots , \mu_N) \, . \end{equation}
For large $N$, the typical distance between neighboring eigenvalues of $H$ is of the order $1/N$. For this reason, non-trivial correlations can only emerge when all arguments of $p^{(k)}$ range within an interval of size of order $1/N$; in this case we speak about local correlations. Using the explicit expression (\ref{eq:GUE}), Dyson showed in \cite{Dy} that the local correlations of GUE in the limit of large N converge to the determinental process associated with the sine-kernel. The next theorem, proven in \cite{ERSTVY}, shows that the same local correlations are observed for general Wigner matrices.

\begin{theorem}\label{thm:uni}
Suppose $H$ is a Wigner matrix as defined in Def. \ref{def}, with $\E \, x_{ij}^3 = 0$. Then, for any fixed $|E| <2$ and $k \in \bN$, we have
\[ \frac{1}{\rho_{sc}^k (E)} \, p^{(k)} \left( E + \frac{x_1}{N \rho_{sc} (E)} , \dots , E+ \frac{x_k}{N \rho_{sc} (E)} \right) \to \left( \frac{\sin (\pi (x_i -x_j))}{ \pi (x_i - x_j)} \right)_{1 \leq i,j \leq k} \] as $N \to \infty$, where $p^{(k)}$ is the $k$-point correlation function defined in (\ref{eq:kpoint}). Here convergence holds in a weak sense, after integrating against a bounded and compactly supported observable $O(x_1, \dots , x_k)$.
\end{theorem}

Without the technical requirement $\E \, x_{ij}^3 = 0$, the same result is proven to hold if one integrates also the variable $E$ over an arbitrarily small but fixed interval; see \cite{ERSTVY}.

Theorem \ref{thm:uni} was obtained combining techniques developed separately in \cite{EPRSY} and \cite{TV}. In \cite{EPRSY}, universality was proven for ensembles of Wigner matrices whose entries have a sufficiently regular distribution. In \cite{TV}, universality was then shown under the assumption that the entries have a vanishing third moment and are supported on at least three points. Both approaches made use of a  previous partial result obtained in \cite{J}, where universality was proven for Wigner matrices of the form $H = H_0 + s V$, with $H_0$ an arbitrary Wigner matrix, $V$ a GUE matrix independent of $H_0$ and $s>0$.

Observe that the technique of \cite{J} does not extend to matrices with different symmetry (real symmetric of quaternion hermitian). For this reason, \cite{EPRSY} does not apply to Wigner matrices with non-hermitian symmetry. Also the result of \cite{TV} only implies universality for real symmetric and quaternion hermitian matrices, if the first four moments of the entries match those of the corresponding gaussian ensemble. A different approach was later proposed in \cite{ESY4,ESYY,EYY}, where universality is established for arbitrary Wigner matrices, independently of their symmetry, after integrating the variable $E$ over small intervals.

\section{Wegner estimate}
\label{sec:weg}

In this section, we establish a new Wegner estimate for the eigenvalues of hermitian ensembles of Wigner matrices. A Wegner estimate is an upper bound for the probability to find an eigenvalue in some interval $I \subset \bR$ which is proportional to the length $|I|$ of the interval, and holds for arbitrarily small intervals, uniformly in $N$.

\begin{theorem}\label{thm:weg}
Let $H$ be a hermitian Wigner matrix, as in Def. \ref{def}. We assume that the random variables $\{ x_{ij}, y_{ij} \}_{1\leq i  < j \leq N}$ have the common probability density function $h$ with the property
\begin{equation}\label{eq:ass-weg} \int \left( \frac{h' (s)}{h(s)} \right)^4 h(s) ds < \infty \, .\end{equation}
For any fixed $\kappa >0$, there exists a constant $C>0$ such that
\begin{equation}\label{eq:weg} \P \left( \cN \left[E-\frac{\eps}{2N} ; E + \frac{\eps}{2N} \right] \geq 1 \right) \leq C \eps \end{equation}
for all $|E| \leq 2 - \kappa$, for all $N \geq 9$ and all $\eps >0$.
\end{theorem}

The result of this theorem and the proof presented below extend easily to ensembles of Wigner matrices with different symmetries (real symmetric and  quaternion hermitian ensembles).

It is important to observe that some regularity of the probability density function of the matrix entries is required for (\ref{eq:weg}) to be correct. If entries have discrete distributions, the event that an eigenvalue lies exactly at $E$ may have non-zero probability. This probability is certainly (exponentially) small in $N$, but it does not depend on $\eps$; hence, in this case, the bound (\ref{eq:weg}) cannot hold uniformly in $\eps >0$.

A Wegner estimate for eigenvalues of Wigner matrices was previously obtained in \cite{ESY3} (see Theorem~3.4) under the assumption that the entries have a probability density function
$h = e^{-g}$ with
\begin{equation}\label{eq:old-weg} |\hat{h} (p)|, |\widehat{hg''} (p)| \leq \frac{1}{(1+C p^2)^{\sigma}} \end{equation} for some $\sigma \geq 5$. To show (\ref{eq:weg}) we combine some new ideas (see, in particular, Lemma \ref{lm:Exi}) with the same general strategy used in \cite{ESY3}. As in \cite{ESY3}, one of the main ingredients in the proof of (\ref{eq:weg}) is the convergence of the density of states on microscopic intervals. More precisely, Theorem \ref{thm:micro} is used to establish the absence of large gaps in the spectrum, as stated in the next Theorem.
\begin{theorem}[Theorem 3.3 of \cite{ESY3}] \label{thm:gap}
Let $H$ be a hermitian Wigner matrix as defined in Def. \ref{def}. Let $\mu_1 \leq \dots \leq \mu_N$ denote the eigenvalues of $H$. Fix $\kappa>0$ and $|E| < 2-\kappa$. Let the (random) index $\alpha$ be so that $\mu_\alpha$ is the largest eigenvalue below
$E$. Then  there are positive constants $C$ and $c$, depending on $\kappa$,
such that
\be
  \P \Big(N (\mu_{\alpha +1} -E) \ge K \quad \text{and }  \; \alpha \leq N-1\Big)
\leq C\; e^{-c\sqrt{K}}
\label{gapdec}
\ee
for any $N\ge 1$ and any $K\ge 0$.
\end{theorem}

With the bound (\ref{gapdec}) we are now ready to prove Theorem \ref{thm:weg}.

\begin{proof}[Proof of Theorem \ref{thm:weg}]
Let $\cN_\eps = \cN  [E- \eps/2N ; E + \eps /2N]$ denote the number of eigenvalues of $H$ in the interval $[E- \eps/2N ; E + \eps /2N]$. We observe that
\[ \begin{split} \cN_\eps = \; \sum_{\alpha=1}^N {\bf 1} (|\mu_\alpha - E| \leq \eps/2N) \lesssim  \sum_{\alpha=1}^N \frac{(\eps/N)^2}{(\mu_\alpha - E)^2 + (\eps/N)^2}  \lesssim \; \frac{\eps}{N} \,  \text{Im } \sum_{\alpha=1}^N \frac{1}{\mu_\alpha - E - i \frac{\eps}{N}}  \end{split} \]
where $\mu_1, \dots , \mu_N$ are the eigenvalues of $H$. Hence
\[  \cN_\eps \lesssim \frac{\eps}{N} \text{Im } \tr \, \frac{1}{H-E-i\frac{\eps}{N}} = \frac{\eps}{N} \text{Im } \sum_{j=1}^N  \left(\frac{1}{H-E-i\frac{\eps}{N}} \right) (j,j) \, . \]
To estimate the right hand side of the last equation, we use that, for any $z \in \bC \backslash \bR$,
\[ \frac{1}{H-z} (j,j) = \frac{1}{h_{jj} - z - \bba^{(j)} \cdot (B^{(j)}-z)^{-1} \bba^{(j)}} \]
where $\bba^{(j)} = (h_{j1}, \dots, h_{j,j-1}, h_{j,j+1}, \dots h_{jN})$ is the $j$-th row of $H$, after removing the diagonal entry $h_{jj}$, while $B^{(j)}$ is the $(N-1) \times (N-1)$ minor of $H$ obtained by removing the $j$-th row and the $j$-th column. We conclude that
\[ \cN_{\eps} \lesssim  \frac{\eps}{N} \sum_{j=1}^N  \frac{1}{\left| h_{jj} - E - i\frac{\eps}{N} - \bba^{(j)} \cdot (B^{(j)}-z)^{-1} \bba^{(j)} \right|} \, . \]
Therefore,
\[ \begin{split}  \P (\cN_\eps \geq 1) \leq \; &\E \, \cN_\eps^2 \\ \lesssim \; & \eps^2 \, \E \, \left[ \frac{1}{N} \sum_{j=1}^N \frac{1}{\left| h_{jj} - E - i\frac{\eps}{N} - \bba^{(j)} \cdot (B^{(j)}-z)^{-1} \bba^{(j)} \right|} \right]^2 \\ \leq \; & \frac{\eps^2}{N} \, \E \, \sum_{j=1}^N \frac{1}{\left| h_{jj} - E - i\frac{\eps}{N} - \bba^{(j)} \cdot (B^{(j)}-z)^{-1} \bba^{(j)} \right|^2}
\\ \leq \; &\eps^2 \, \E \, \frac{1}{\left| h_{11} - E - i\frac{\eps}{N} - \bba^{(1)} \cdot (B^{(1)}-z)^{-1} \bba^{(1)} \right|^2} \end{split} \]
where we used the convexity of $x \to x^2$ and, in the last line, the symmetry w.r.t. permutations of rows of $H$. Next, let $\lambda_1, \dots , \lambda_{N-1}$ and $\bu_1, \dots , \bu_{N-1}$ be the eigenvalues of the minor $B^{(1)}$ and the corresponding (normalized) eigenvectors. Moreover, let $\bb = (b_1, \dots ,b_{N-1}) = \sqrt{N} \bba^{(1)}$ (the factor $\sqrt{N}$ compensates for the scaling of the matrix entries introduced in Def. \ref{def}; the random variables $b_j$, $j=1, \dots , N-1$ are so that $\E \, b_j =0$ and $\E \, |b_j|^2 = 1/2$). We find
\begin{equation}\label{eq:weg1} \begin{split} \P (\cN_\eps \geq 1) \lesssim \; & \eps^2 \, \E \, \frac{1}{\left| h_{11} - E - i\frac{\eps}{N} - \frac{1}{N} \sum_{\alpha=1}^{N-1} \frac{|\bb \cdot \bu_\alpha|^2}{\lambda_\alpha - E - i \frac{\eps}{N}}\right|^2} \\ \leq \; & \eps^2 \, \E \, \frac{1}{\left(h_{11} -E - \sum_{\alpha=1}^{N-1} d_\alpha |\bb \cdot \bu_\alpha|^2 \right)^2 + \left( \frac{\eps}{N} + \sum_{\alpha=1}^{N-1} c_\alpha |\bb \cdot \bu_\alpha|^2 \right)^2} \end{split} \end{equation}
where we defined the coefficients
\begin{equation}\label{eq:cdalpha} c_\alpha = \frac{\eps}{N^2 (\lambda_\alpha - E)^2 + \eps^2} \qquad \text{and } \quad
d_\alpha  = \frac{N (\lambda_\alpha - E)}{N^2 (\lambda_\alpha - E)^2 + \eps^2} \, . \end{equation}
It is important to notice that the entries $b_j$, $j=1,\dots , N-1$ are independent of the eigenvalues $\lambda_\alpha$ and the eigenvectors $\bu_\alpha$ of $B^{(1)}$. For this reason, we can compute the expectation in \eqref{eq:weg1} by first averaging over the vector $\bb$, keeping the randomness associated with $B^{(1)}$ (in particular, the coefficients $c_\alpha$ and $d_\alpha$ defined in (\ref{eq:cdalpha}) and the eigenvectors $\bu_\alpha$) fixed.

We define $\Omega$ to be the event that at least six eigenvalues of $B^{(1)}$ are located outside the interval $[E-\eps /2N ; E+ \eps/2N]$. On the set $\Omega$ and on its complement $\Omega^c$, we derive different bounds for the expectation over the random vector $\bb$. We write
\begin{equation} \label{eq:weg-2} \begin{split}
 \P (\cN_\eps \geq 1) \lesssim \; & \eps^2 \, \E_B \, {\bf 1}_\Omega \, \E_\bb \, \frac{1}{\left(h_{11} -E - \sum_{\alpha=1}^{N-1} d_\alpha |\bb \cdot \bu_\alpha|^2 \right)^2 + \left( \frac{\eps}{N} + \sum_{\alpha=1}^{N-1} c_\alpha |\bb \cdot \bu_\alpha|^2 \right)^2}
\\ &+ \eps^2 \, \E_B \, {\bf 1}_{\Omega^c} \, \E_\bb \, \frac{1}{\left(h_{11} -E - \sum_{\alpha=1}^{N-1} d_\alpha |\bb \cdot \bu_\alpha|^2 \right)^2  + \left( \frac{\eps}{N} + \sum_{\alpha=1}^{N-1} c_\alpha |\bb \cdot \bu_\alpha|^2 \right)^2} \\
=: \; & \text{A} + \text{B}
\end{split} \end{equation}
where $\E_B$ denotes the expectation over the randomness associated with the minor $B^{(1)}$, while $\E_\bb$ denotes the expectation over the vector $\bb$.

In the exceptional set $\Omega^c$, we can find (since $N \geq 8$), indices $\beta_1, \beta_2 , \beta_3 \in \{ 1, \dots , N-1\}$ such that $\lambda_{\beta_j} \in [E - \eps/2N ; E + \eps / 2N]$ for $j=1,2,3$. Then $c_{\beta_j} > 1/(2\eps)$ for $j=1,2,3$. Therefore, the second term on the r.h.s. of (\ref{eq:weg-2}) is bounded by
\[ \begin{split}
\text{B}  \leq \; &\eps^2 \,  \E_B  \, {\bf 1}_{\Omega^c} \,  \E_\bb \, \frac{1}{\left( \sum_{j=1}^3 c_{\beta_j} |\bb \cdot \bu_{\beta_j}|^2 \right)^2} \\ \lesssim \; &\eps^4  \sup_{\bu_1, \bu_2, \bu_3} \, \E_\bb \, \frac{1}{\left( \sum_{j=1}^3  |\bb \cdot \bu_j|^2 \right)^2}
\end{split} \]
where the supremum is taken over all sets $\{ \bu_1, \bu_2 , \bu_3 \}$ of three orthonormal vectors in $\bC^{N-1}$. Lemma~\ref{lm:Exi} implies that $\text{B} \lesssim \eps^4$.

Next, we focus on the first term on the r.h.s. of (\ref{eq:weg-2}). On the set $\Omega$, we can define indices $\alpha_1, \dots , \alpha_6$ as follows. We fix $\alpha_1 \in \{ 1, \dots , N-1\}$ so that
\[ |\lambda_{\alpha_1} - E| = \inf \{ |\lambda_\alpha - E| : |\lambda_\alpha - E| \geq \eps /N \} \, . \]  Moreover, we choose recursively the indices $\alpha_2, \dots , \alpha_6 \in \{ 1, \dots , N-1\}$ by the formula
\[ |\lambda_{\alpha_j} - E| = \inf \{ |\lambda_\alpha - E| : |\lambda_\alpha - E| \geq \eps /N, \alpha_j \not = \alpha_i, \text{for all } 1 \leq i < j \} \, .\]
We define $\Delta = N |\lambda_{\alpha_6} - E|$. Then \[ \eps < N |\lambda_{\alpha_j} - E| \leq \Delta, \] for every $1 \leq j \leq 6$.  This implies that
\begin{equation}\label{eq:dalpha} |d_{\alpha_4}| > |d_{\alpha_5}| > |d_{\alpha_6}| > \frac{1}{2\Delta} \end{equation}
while
\begin{equation}\label{eq:calpha} c_{\alpha_1} > c_{\alpha_2} > c_{\alpha_3} > \frac{\eps}{2\Delta^2} \, . \end{equation}
{F}rom (\ref{eq:weg-2}), we conclude that
\begin{equation*} \begin{split}
\text{A} \lesssim\; & \eps^2 \, \E_{B} \, {\bf 1}_\Omega \, E_{\bb} \,  \frac{1}{(h_{11} -E-\sum_\alpha d_\alpha |\bb \cdot \bu_\alpha|^2)^2 + (\sum_{j=1}^3 c_{\alpha_j} |\bb \cdot \bu_{\alpha_j}|^2)^2} \\ \leq \; & \eps^2 \, \E_{B} \, {\bf 1}_\Omega \, \int \prod_{j=1}^{N-1}  db_j \, d \overline{b}_j  \, h (\text{Re } b_j) \, h( \text{Im } b_j)   \frac{1}{(h_{11} -E-\sum_\alpha d_\alpha |\bb \cdot \bu_\alpha|^2)^2 + (\sum_{j=1}^3 c_{\alpha_j} |\bb \cdot \bu_{\alpha_j}|^2)^2} \,. \end{split} \end{equation*}

Now we define
\[ F(t) = \int_{-\infty}^t ds \, \frac{1}{s^2 + \left( \sum_{j=1}^3 c_{\alpha_j} |\bb \cdot \bu_{\alpha_j} |^2 \right)^2} \]
Then we have
\begin{equation}\label{eq:Fbd} 0 \leq F(t) \leq \frac{1}{\sum_{j=1}^3 c_{\alpha_j} |\bb \cdot \bu_{\alpha_j}|^2} \end{equation}
for every $t \in \bR$. Moreover, with the notation $\sigma_j = 1$ if $\lambda_{\alpha_j} > E$ and $\sigma_j = -1$ if $\lambda_{\alpha_j} < E$, for $j =1, \dots ,6$, we find
\[ \begin{split}
\sum_{i=4}^6 \sigma_i (\bu_{\alpha_i} \cdot \bb) \sum_{\ell} \bu_{\alpha_i} (\ell) &\frac{d}{db_\ell} \, F(h_{11}-E-\sum_\alpha d_\alpha |\bb \cdot \bu_\alpha|^2) \\ &= \frac{ \sum_{i=4}^6 |d_{\alpha_i}| |\bb \cdot \bu_{\alpha_i}|^2}{(h_{11}-E-\sum_\alpha d_\alpha |\bb \cdot \bu_\alpha|^2)^2 + (\sum_{j=1}^3 c_{\alpha_j} |\bb \cdot \bu_{\alpha_j}|^2)^2} \end{split} \]
and therefore
\[ \begin{split}
\text{A} \lesssim \;  \eps^2 \sum_{i=4}^6 \, \E_B  \, {\bf 1}_\Omega \, \sigma_i  \int \prod_{j=1}^{N-1} db_j \, d \overline{b}_j \, &h (\text{Re } b_j) \, h(\text{Im } b_j) \,  \frac{(\bu_{\alpha_i} \cdot \bb)}{\sum_{i=4}^6 |d_{\alpha_i}| \, |\bb \cdot \bu_{\alpha_i}|^2}  \\ &\times \sum_{\ell} \bu_{\alpha_i} (\ell) \frac{d}{db_\ell} \, F(h_{11}-E-\sum_\alpha d_\alpha |\bb \cdot \bu_\alpha|^2) \,.\end{split} \]
Integration by parts gives
\begin{equation}\label{eq:2}  \begin{split}
\text{A} \lesssim \; & \eps^2 \, \E_B   \, {\bf 1}_\Omega \, \int \prod_{j=1}^{N-1} db_j \, d \overline{b}_j \,  h (\text{Re } b_j) \, h (\text{Im } b_j)  \frac{1}{\sum_{i=4}^6 |d_{\alpha_i}| \, |\bb \cdot \bu_{\alpha_i}|^2} \, F(h_{11}-E-\sum_\alpha d_\alpha |\bb \cdot \bu_\alpha|^2)  \\ &
+ \eps^2 \, \E_B  \, {\bf 1}_\Omega \,  \sum_{i=4}^6 \int \prod_{j=1}^{N-1} db_j \, d \overline{b}_j \, h (\text{Re } b_j) \, h (\text{Im } b_j) \, \left| \sum_{\ell} \bu_{\alpha_i} (\ell) \, \left(\frac{h' (\text{Re } b_\ell)}{h(\text{Re } b_\ell)} - i \frac{h' (\text{Im } b_\ell)}{h(\text{Im } b_\ell)} \right) \right|   \\ & \hspace{3cm} \times  \frac{ (\bu_{\alpha_i} \cdot \bb)}{\sum_{i=4}^6 |d_{\alpha_i}| \, |\bb \cdot \bu_{\alpha_i}|^2} \, F(h_{11}-E-\sum_\alpha d_\alpha |\bb \cdot \bu_\alpha|^2) \\ =\; & \text{I } + \text{II} \, . \end{split} \end{equation}
For the first term, we obtain, from (\ref{eq:Fbd}),
\begin{equation}\label{eq:Ibd} \begin{split} \text{I } \lesssim \; &\eps^2 \,  \E_B \, \frac{ {\bf 1}_\Omega }{(\min_{4 \leq j \leq 6} |d_{\alpha_j}|)(\min_{1\leq i \leq 3} c_{\alpha_i} )} \, \E_{\bb} \frac{1}{\sum_{j=4}^6 |\bb \cdot \bu_{\alpha_j}|^2} \frac{1}{\sum_{i=1}^3 |\bb \cdot \bu_{\alpha_i}|^2} \\ \lesssim \; & \eps \, \E_B \, {\bf 1}_\Omega  \, \Delta^3 \, \left( \E_{\bb} \frac{1}{ (\sum_{j=4}^6 |\bb \cdot \bu_{\alpha_j}|^2)^2} \right)^{1/2} \, \left( \E_{\bb}  \frac{1}{ (\sum_{i=1}^3 |\bb \cdot \bu_{\alpha_i}|^2)^2} \right)^{1/2} \\ \lesssim \; & \eps \,  \E_B  \, {\bf 1}_\Omega \, \Delta^3 \end{split}\end{equation}
where we used Lemma \ref{lm:Exi} and the assumption (\ref{eq:ass-weg}).

As for the second term on  the r.h.s. of (\ref{eq:2}) we find,  using H\"older's inequality,
\begin{equation}\label{eq:3} \begin{split}
\text{II } \leq \; &\eps^2 \sum_{i=5}^8 \E_B  \, {\bf 1}_\Omega \, \left( \int \prod_{j=1}^{N-1} db_j \, d \overline{b}_j \, h (\text{Re } b_j) \, h(\text{Im } b_j) \, \left| \sum_{\ell=1}^{N-1} \bu_{\alpha_i} (\ell) \, \left(\frac{h' (\text{Re } b_\ell)}{h(\text{Re } b_\ell)} - i \frac{h' (\text{Im } b_\ell)}{h(\text{Im } b_\ell)} \right) \right|^4 \right)^{1/4} \\ & \hspace{.2cm} \times \left( \int \prod_{j=1}^{N-1}  db_j \, d \overline{b}_j \, h (\text{Re } b_j) h (\text{Im } b_j) \, \frac{|\bu_{\alpha_i} \cdot \bb|^{4/3}}{\left(\sum_{i=4}^6 |d_{\alpha_i}| \, |\bb \cdot \bu_{\alpha_i}|^2\right)^{4/3}} \, \frac{1}{\left(\sum_{i=1}^3 c_{\alpha_i} |\bb\cdot \bu_{\alpha_i}|^2\right)^{4/3}} \right)^{3/4} \, . \end{split} \end{equation}
Now we observe that
\[ \begin{split}  \int \prod_{j=1}^{N-1} &db_j \, d \overline{b}_j \, h (\text{Re } b_j) h (\text{Im } b_j) \, \left| \sum_{\ell} \bu_{\alpha_i} (\ell) \, \left(\frac{h' (\text{Re } b_\ell)}{h(\text{Re } b_\ell)} - i \frac{h' (\text{Im } b_\ell)}{h(\text{Im } b_\ell)} \right) \right|^4 \\
&= \sum_{\ell_1,\ell_2, \ell_3, \ell_4 = 1}^{N-1} \bu_{\alpha_i} (\ell_1) \bu_{\alpha_i} (\ell_2) \overline{\bu}_{\alpha_i} (\ell_3) \overline{\bu}_{\alpha_i} (\ell_4) \\ & \hspace{.2cm} \times  \int  \prod_{j=1}^{N-1}  db_j \, d \overline{b}_j \, h (\text{Re } b_j) h (\text{Im } b_j) \left( \frac{h' (\text{Re } b_{\ell_1})}{h(\text{Re } b_{\ell_1})} - i\frac{h' (\text{Im } b_{\ell_1})}{h(\text{Im } b_{\ell_1})} \right) \, \left(  \frac{h' (\text{Re } b_{\ell_2})}{h(\text{Re } b_{\ell_2})} - i\frac{h' (\text{Im } b_{\ell_2})}{h(\text{Im } b_{\ell_2})} \right) \\ &\hspace{5cm} \times \left( \frac{h' (\text{Re } b_{\ell_3})}{h(\text{Re } b_{\ell_3})} + i\frac{h' (\text{Im } b_{\ell_3})}{h(\text{Im } b_{\ell_3})} \right) \, \left(  \frac{h' (\text{Re } b_{\ell_4})}{h(\text{Re } b_{\ell_4})} + i\frac{h' (\text{Im } b_{\ell_4})}{h(\text{Im } b_{\ell_4})} \right) \, .
\end{split} \]
Since $\int h' (s) ds = 0$, only terms with $\ell_1 = \ell_3$ and $\ell_2 = \ell_4$ or with $\ell_1 = \ell_4$ and $\ell_2= \ell_3$ do not vanish. Hence
\[ \begin{split}  \int \prod_{j=1}^{N-1} db_j \, d \overline{b}_j \, h (\text{Re } b_j) h (\text{Im } b_j) \,& \left| \sum_{\ell} \bu_{\alpha_i} (\ell) \, \left(\frac{h' (\text{Re } b_\ell)}{h(\text{Re } b_\ell)} - i \frac{h' (\text{Im } b_\ell)}{h(\text{Im } b_\ell)} \right) \right|^4
\\ &\lesssim  \left[ \left(\int ds \, h(s) \, \left(\frac{h'(s)}{h(s)}\right)^2 \right)^2 + \int ds \, h(s) \left( \frac{h' (s)}{h(s)} \right)^4 \right] \, .  \end{split} \]
Using the assumption (\ref{eq:ass-weg}), we conclude from (\ref{eq:3}) that
\[ \begin{split}
\text{II } \lesssim  \; &\eps \,  \E_B   \, {\bf 1}_\Omega \, \Delta^3  \left( \int \prod_{j=1}^{N-1} db_j \, d \overline{b}_j \, h (\text{Re } b_j) \, h (\text{Im } b_j)  \frac{1}{\left(\sum_{i=4}^6 \, |\bb \cdot \bu_{\alpha_i}|^2\right)^{2/3}} \, \frac{1}{\left(\sum_{i=1}^3 |\bb\cdot \bu_{\alpha_i}|^2\right)^{4/3}} \right)^{3/4} \\  \leq \; &\eps \, \E_B  \, {\bf 1}_\Omega  \, \Delta^3 \left( \E_{\bb} \frac{1}{(\sum_{i=4}^6 |\bb \cdot \bu_{\alpha_i}|^2)^2} \right)^{1/4} \left( \E_{\bb} \frac{1}{(\sum_{j=1}^3 |\bb \cdot \bu_{\alpha_j}|^2)^2} \right)^{1/2} \\ \leq \; &\eps \, \E_B  \, {\bf 1}_\Omega \, \Delta^3 \end{split} \]
where, in the last line, we used Lemma \ref{lm:Exi}. Last equation, combined with  (\ref{eq:Ibd}) and (\ref{eq:2}), implies that
\[ \text{A} \lesssim \eps \, \E_B  \, {\bf 1}_\Omega \, \Delta^3 \, . \]
{F}rom Theorem \ref{thm:gap}, we have
\[ \P (\Delta > K \text{ and } \Omega) \lesssim e^{-c \sqrt{K}} \, . \] Therefore, $\E  \, {\bf 1}_\Omega \, \Delta^3$ is finite (uniformly in $N$ and $\eps >0$) and we conclude that $\text{A} \lesssim \eps$, and therefore, from (\ref{eq:weg-2}), that
\[  \P (\cN_\eps \geq 1) \lesssim \eps \, . \]
\end{proof}

The next lemma is the main new ingredient compared with to the proof presented in \cite{ESY3}.

\begin{lemma}\label{lm:Exi}
Let $\bb = (b_1, \dots, b_{N-1}) \in \bC^{N-1}$, where $\{ \text{Re } b_j , \text{Im } b_j \}_{j=1}^{N-1}$ is a collection of $2(N-1)$ independent and identically distributed random variables with a common probability density function $h$ such that
\begin{equation}\label{eq:ass-h1} \int ds \,h(s) \,  \left(\frac{h'(s)}{h(s)} \right)^4  < \infty \, \qquad \text{and } \quad \int ds \,  s^4 \, h(s) < \infty \, . \end{equation}
Fix $r=1,2$, $m \in \bN$ with $m >r$. Then there exists a constant $C>0$ such that  \begin{equation}\label{eq:Exi} \E \; \frac{1}{\left(\sum_{j=1}^m |\bb \cdot \bu_j|^2 \right)^r} < C \end{equation} for any $N \geq 2$ and for any set of $m$ orthonormal vectors $\bu_1, \dots , \bu_m \in \bC^{N-1}$ .
\end{lemma}
\begin{proof}
{F}rom the monotone convergence theorem, we obtain
\[ \begin{split}
\E \; \frac{1}{\left(\sum_{j=1}^m |\bb \cdot \bu_j|^2 \right)^r} = \; & \lim_{\delta \to 0} \int
\left[ \prod_{i=1}^{N-1} db_i d \overline{b}_i \, h (\text{Re } b_i) h(\text{Im } b_i) \right] e^{-\delta \sum_{\ell=1}^{N-1} |b_\ell|^2} \, \frac{1}{\left(\sum_{j=1}^m |\bb \cdot \bu_j|^2 \right)^r} \\
= \; & \lim_{\delta \to 0^+} \int
\left[ \prod_{i=1}^{N-1} db_i d \overline{b}_i \, h_\delta (\text{Re } b_i) h_\delta (\text{Im } b_i) \right] \,\frac{1}{\left(\sum_{j=1}^m |\bb \cdot \bu_j|^2 \right)^r}
\end{split} \]
where we introduced the notation $h_\delta (x) = h (x) e^{-\delta x^2}$. Observe that
\begin{equation}\label{eq:heps} 0 \leq h_\delta (x) \leq C e^{-\delta x^2} \, . \end{equation}
This follows from (\ref{eq:ass-h1}) because
\[ \| h \|_{\infty} = \| \sqrt{h} \|^2_{\infty} \lesssim \| \sqrt{h} \|^2_{H^1}  = \int h (s) ds + \frac{1}{4} \int \frac{(h' (s))^2}{h(s)} ds  \lesssim 1+  \int \left( \frac{h' (s)}{h(s)} \right)^4 \, h(s) ds \, \]
where, in the last step, we used the fact that $h (s) ds$ is a probability measure. Using again the monotone convergence theorem, we obtain
\[ \E \; \frac{1}{\left(\sum_{j=1}^m |\bb \cdot \bu_j|^2 \right)^r} =  \lim_{\delta \to 0^+} \lim_{\kappa \to 0^+} \int \left[
\prod_{i=1}^{N-1} db_i d \overline{b}_i \, h_\delta (\text{Re } b_i) h_\delta (\text{Im } b_i) \right] \,\frac{{\bf 1} (\sum_{j=1}^m |\bb \cdot \bu_j|^2 \geq \kappa^2 )}{\left(\sum_{j=1}^m |\bb \cdot \bu_j|^2  \right)^r} \, .
\]
We complete now $\bu_1, \dots ,\bu_m$ to an orthonormal basis $\bu_1, \dots , \bu_{N-1}$ of $\bC^{N-1}$, and we introduce new coordinates $z_\alpha = \bu_\alpha \cdot \bb$, for $\alpha = 1, \dots, N-1$. We also use the notation $\bv = (z_1, \dots, z_m) \in \bC^m$, $\bw = (z_{m+1}, \dots, z_{N-1}) \in \bC^{N-m-1}$, and $\bz = (\bv, \bw) = (z_1, \dots , z_{N-1}) \in \bC^{N-1}$ and we denote $U$ the unitary $(N-1) \times (N-1)$ matrix, with columns $\bu_1, \dots , \bu_{N-1}$ (so that $\bz= U^* \bb$). We have
\begin{equation}\label{eq:Exi2}
\E \; \frac{1}{\left(\sum_{j=1}^m |\bb \cdot \bu_j|^2 \right)^r} =  \lim_{\delta \to 0^+} \lim_{\kappa \to 0^+} \int  d\bz d \overline{\bz} \, \left[ \prod_{i=1}^{N-1} h_\delta (\text{Re } (U\bz)_i) h_\delta (\text{Im } (U\bz)_i) \right] \,\frac{{\bf 1} (|\bv| \geq \kappa)}{|\bv|^{2r}} \, .
\end{equation}
We observe that
\[ \sum_{\ell = 1}^m \left[  \frac{d}{d\text{Re } z_\ell} \, \frac{\text{Re } z_\ell}{
|\bv|^{2r}} + \frac{d}{d\text{Im } z_\ell} \, \frac{\text{Im } z_\ell}{|\bv|^{2r}} \right] = \frac{2(m-r)}{|\bv|^{2r}} \, .  \]
Therefore we have
\begin{equation}\label{eq:zmw} \begin{split}
 \int  d\bz d \overline{\bz} \, & \left[ \prod_{i=1}^{N-1}  \, h_\delta (\text{Re } (U\bz)_i) h_\delta (\text{Im } (U\bz)_i)  \right] \,\frac{{\bf 1} (|\bv| \geq \kappa)}{|\bv|^{2r}} \\ = \; & \frac{1}{2(m-r)}\sum_{\ell=1}^m \int d\bw d \overline{\bw} \\ &\hspace{.5cm} \times \int_{|\bv| \geq \kappa}  d\bv d \overline{\bv} \left[ \prod_{i=1}^{N-1} h_\delta (\text{Re } (U\bz)_i) h_\delta (\text{Im } (U\bz)_i) \right] \left[ \frac{d}{d\text{Re } z_\ell} \, \frac{\text{Re } z_\ell}{ |\bv|^{2r}} + \frac{d}{d\text{Im } z_\ell} \, \frac{\text{Im } z_\ell}{ |\bv|^{2r}} \right]
 \\ = \; & \frac{1}{2(m-r)} \sum_{\ell=1}^m \int d\bw d \overline{\bw} \\ &\hspace{.5cm}\times \int_{|\bv| \geq \kappa}  d\bv d \overline{\bv}  \left[ \frac{d}{d \text{Re } z_\ell} \,\left(\frac{\text{Re } z_\ell}{ |\bv|^{2r}}  \prod_{i=1}^{N-1} h_\delta (\text{Re } (U\bz)_i) h_\delta (\text{Im } (U\bz)_i)\right)   \right. \\ & \hspace{4cm} \left. +  \frac{d}{d \text{Im } z_\ell}\, \left(\frac{\text{Im } z_\ell}{ |\bv|^{2r}}\prod_{i=1}^{N-1} h_\delta (\text{Re } (U\bz)_i) h_\delta (\text{Im } (U\bz)_i) \right)\right] \\
&- \frac{1}{2(m-r)} \sum_{\ell =1}^m \int d\bw d \overline{\bw} \int_{|\bv| \geq \kappa}  d\bv d \overline{\bv} \, \prod_{i=1}^{N-1} h_\delta (\text{Re } (U\bz)_i) h_\delta (\text{Im } (U\bz)_i) \\ & \hspace{3cm} \times \frac{\text{Re } z_\ell}{ |\bv|^{2r}} \left( \sum_{j=1}^{N-1} \text{Re } U_{\ell j} \frac{h'_\delta (\text{Re } (U\bz)_j)}{h_\delta (\text{Re } (U\bz)_j)} + \text{Im } U_{\ell j} \frac{h'_\delta (\text{Im } (U\bz)_j)}{h_\delta (\text{Im } (U\bz)_j)} \right)  \\
&- \frac{1}{2(m-r)} \sum_{\ell =1}^m \int d\bw d \overline{\bw} \int_{|\bv| \geq \kappa}  d\bv d \overline{\bv} \, \prod_{i=1}^{N-1} h_\delta (\text{Re } (U\bz)_i) h_\delta (\text{Im } (U\bz)_i) \\ & \hspace{3cm} \times  \frac{\text{Im } z_\ell}{ |\bv|^{2r}} \left( \sum_{j=1}^{N-1} \text{Im } U_{\ell j} \frac{h'_\delta (\text{Re } (U\bz)_j)}{h_\delta (\text{Re } (U\bz)_j)} - \text{Re } U_{\ell j} \frac{h'_\delta (\text{Im } (U\bz)_j)}{h_\delta (\text{Im } (U\bz)_j)} \right) \\
= \; & \text{I} + \text{II} + \text{III} \, .  \end{split} \end{equation}
Using Gauss's Divergence Theorem and (\ref{eq:heps}), the first term is bounded by
 \begin{equation}\label{eq:Exi-I} \begin{split}
\text{I}   & =  \frac{1}{2(m-r)} \int d\bw d \overline{\bw}  \int_{|\bv| = \kappa} d\bv d\overline{\bv}\,  \frac{1}{ |\bv|^{2r-1}}\, \prod_{\alpha=1}^{N-1} h_\delta (\text{Re } (U\bz)_j) h_\delta (\text{Im } (U\bz)_j)   \\ & \leq C^N \kappa^{-2r+1} \int d\bw d \overline{\bw} \,  e^{-\delta \bw^2} \int_{|\bv| = \kappa} d\bv d\overline{\bv}
\\ & \leq C^N \kappa^{2(m-r)} \int d\bw d \overline{\bw} \,  e^{-\delta \bw^2}  \\ &\leq C_{N,\delta} \, \kappa^{2(m-r)}
 \end{split} \end{equation}
where $C_{N,\delta}$ is a constant depending on $N$ and $\delta$ (and also $m,r$). Since $m >r$, this contribution vanishes in the limit $\kappa \to 0$.

\medskip

For $r=1$, the second term on the r.h.s. of (\ref{eq:zmw}) is bounded, for any constant $\alpha >0$, by
\[ \begin{split}
|\text{II}| \leq \; & \alpha \sum_{\ell =1}^m \int d\bz d \overline{\bz} \, \prod_{i=1}^{N-1} h_\delta (\text{Re } (U\bz)_i) h_\delta (\text{Im } (U\bz)_i) \\ &\hspace{2cm} \times \left( \sum_{j=1}^{N-1} \text{Re } U_{\ell j} \, \frac{h'_\delta (\text{Re } (U\bz)_j)}{h_\delta (\text{Re } (U\bz)_j)} + \text{Im } U_{\ell j} \, \frac{h'_\delta (\text{Im } (U\bz)_j)}{h_\delta (\text{Im } (U\bz)_j)} \right)^2 \\ &+  \alpha^{-1} \sum_{\ell =1}^m \int d\bz d \overline{\bz} \, \prod_{i=1}^{N-1} h_\delta (\text{Re } (U\bz)_i) h_\delta (\text{Im } (U\bz)_i) \frac{(\text{Re } z_\ell)^2}{|\bv|^4} \\  \lesssim \; &\alpha \sum_{\ell=1}^m \sum_{j_1, j_2 =1}^{N-1}\int d\bb d\overline{\bb} \prod_{i=1}^{N-1} h_\delta (\text{Re } b_i) h_\delta (\text{Im } b_i) \\ &\hspace{2cm} \times \left( \text{Re } U_{\ell j_1} \, \frac{h'_\delta (\text{Re } (U\bz)_{j_1})}{h_\delta (\text{Re } (U\bz)_{j_1})} + \text{Im } U_{\ell j_1} \, \frac{h'_\delta (\text{Im } (U\bz)_{j_1})}{h_\delta (\text{Im } (U\bz)_{j_1})} \right) \\ & \hspace{4cm} \times  \left( \text{Re } U_{\ell j_2} \, \frac{h'_\delta (\text{Re } (U\bz)_{j_2})}{h_\delta (\text{Re } (U\bz)_{j_2})} + \text{Im } U_{\ell j_2} \, \frac{h'_\delta (\text{Im } (U\bz)_{j_2})}{h_\delta (\text{Im } (U\bz)_{j_2})} \right) \\ &+  \alpha^{-1} \sum_{\ell =1}^m \int d\bz d \overline{\bz} \, \prod_{i=1}^{N-1} h_\delta (\text{Re } (U\bz)_i) h_\delta (\text{Im } (U\bz)_i) \, \frac{(\text{Re } z_\ell)^2}{|\bv|^4}
\\
\lesssim \; & \alpha  \int  \, \left( \frac{h'_\delta (s)}{h_\delta (s)} \right)^2 h_\delta (s) ds  +  \alpha^{-1} \, \E \frac{1}{\sum_{j=1}^m |\bb \cdot \bu_j|^2 } \, .  \end{split}\]
In the last line, we used the fact that terms with $j_1 \not = j_2$ vanish (because $\int h'_{\delta} (s) ds = 0$) (note that $A \lesssim B$ means here that $A \leq c B$, where the constant $c$ may depend only on $m$ and $r$).

For $r=2$, the second term in (\ref{eq:zmw}) can be bounded similarly by
\[  \begin{split}
|\text{II}| \leq \; & \alpha \sum_{\ell =1}^m \int d\bz d \overline{\bz} \, \prod_{i=1}^{N-1} h_\delta (\text{Re } (U\bz)_i) h_\delta (\text{Im } (U\bz)_i) \\ &\hspace{2cm} \times  \left( \sum_{j=1}^{N-1} \text{Re } U_{\ell j} \frac{h'_\delta (\text{Re } (U\bz)_j)}{h_\delta (\text{Re } (U\bz)_j)} + \text{Im } U_{\ell j} \frac{h'_\delta (\text{Im } (U\bz)_j)}{h_\delta (\text{Im } (U\bz)_j)} \right)^4 \\ &+  \alpha^{-1} \sum_{\ell =1}^m \int d\bz d \overline{\bz} \, \prod_{i=1}^{N-1} h_\delta (\text{Re } (U\bz)_i) h_\delta (\text{Im } (U\bz)_i) \frac{|\text{Re } z_\ell |^{4/3}}{|\bv|^{16/3}} \\  \lesssim \; & \alpha \, \left[ \left(  \int  \left(\frac{h'_\delta (s)}{h_\delta (s)} \right)^2 h_\delta (s) ds \right)^2 + \int \, \left(\frac{h'_\delta (s)}{h_\delta (s)} \right)^4 h_\delta (s) ds \right]  +  \alpha^{-1} \, \E \frac{1}{\left(\sum_{j=1}^m |\bb \cdot \bu_j|^2 \right)^2 } \, .  \end{split}\]
Observe that, from $h_\delta (s) = e^{-\delta s^2} h (s)$, we have $h'_\delta (s) = -2\delta s e^{-\delta s^2} h(s) + e^{-\delta s^2} h' (s)$ and therefore
\[ \frac{h'_\delta (s)}{h_\delta (s)} = -2\delta s + \frac{h' (s)}{h(s)}\, . \]
This implies that, for any $p>0$,
\[ \int \left( \frac{h'_\delta (s)}{h_\delta (s)} \right)^p \, h_\delta (s) ds \lesssim \delta^p \int s^p h(s) ds + \int  \left( \frac{h' (s)}{h (s)} \right)^p \, h (s) ds \, . \]
{F}rom (\ref{eq:ass-h1}), we have that, for $r=1,2$ and $m\in \bN$ with $m >r$, and or any $\alpha >0$,
\[ | \text{II} | \lesssim \alpha (1 + \delta^4) + \alpha^{-1} \E \, \frac{1}{\left(\sum_{j=1}^m |\bb \cdot \bu_j|^2 \right)^r} \, .\]
The third term on the r.h.s. of (\ref{eq:zmw}) can be bounded in exactly the same way. Hence, with (\ref{eq:Exi-I}) and (\ref{eq:Exi2}), we find
\[ \begin{split}
\E \; \frac{1}{\left(\sum_{j=1}^m |\bb \cdot \bu_j|^2 \right)^r} \lesssim \alpha  + \alpha^{-1} \, \E \; \frac{1}{\left(\sum_{j=1}^m |\bb \cdot \bu_j|^2 \right)^r} \, . \end{split} \]
Choosing $\alpha >0$ sufficiently large, we find
\[ \E \; \frac{1}{\left(\sum_{j=1}^m |\bb \cdot \bu_j|^2 \right)^r} \lesssim 1 \, . \]
\end{proof}

\end{document}